\documentclass[11pt]{article}

\usepackage{amsfonts,amsmath,amsxtra}
\usepackage{amssymb,amscd}
\usepackage[all]{xy}
\usepackage{bm}
\catcode`\@=11
\def\marginnote#1{}
\newcount\hour
\newcount\minute
\newtoks\amorpm
\hour=\time\divide\hour by60
\minute=\time{\multiply\hour by60 \global\advance\minute by-\hour}
\edef\standardtime{{\ifnum\hour<12 \global\amorpm={am}%
     \else\global\amorpm={pm}\advance\hour by-12 \fi
     \ifnum\hour=0 \hour=12 \fi
   \number\hour:\ifnum\minute<10 0\fi\number\minute\the\amorpm}}
\edef\militarytime{\number\hour:\ifnum\minute<10 0\fi\number\minute}
\def\draftlabel#1{{\@bsphack\if@filesw {\let\thepage\relax
\xdef\@gtempa{\write\@auxout{\string
   \newlabel{#1}{{\@currentlabel}{\thepage}}}}}\@gtempa
\if@nobreak \ifvmode\nobreak\fi\fi\fi\@esphack}
     \gdef\@eqnlabel{#1}}
\def\@eqnlabel{}
\def\@vacuum{}
\def\draftmarginnote#1{\marginpar{\raggedright\scriptsize\tt#1}}
\def\draft{\oddsidemargin -0.1truein
     \def\@oddfoot{\sl preliminary draft \hfil
     \rm\thepage\hfil\sl\today\quad\militarytime}
     \let\@evenfoot\@oddfoot \overfullrule 3pt
     \let\label=\draftlabel
     \let\marginnote=\draftmarginnote
\def\@eqnnum{{\rm (\theequation)}
\rlap{\kern\marginparsep\tt\@eqnlabel}%
\global\let\@eqnlabel\@vacuum}  }
\def\numberbysection{\@addtoreset{equation}{section}
     \def\theequation{\thesection.\arabic{equation}}}
\numberbysection

\renewcommand{\theequation}{\thesection.\arabic{equation}}
\makeatletter
\newdimen\normalarrayskip            
\newdimen\minarrayskip               
\normalarrayskip\baselineskip
\minarrayskip\jot
\newif\ifold             \oldtrue            \def\new{\oldfalse}
\def\arraymode{\ifold\relax\else\displaystyle\fi}
\def\eqnumphantom{\phantom{(\theequation)}} 
\def\@arrayskip{\ifold\baselineskip\z@\lineskip\z@
  \else
  \baselineskip\minarrayskip\lineskip1\baselineskip\fi}
\def\@arrayclassz{\ifcase \@lastchclass \@acolampacol \or
\@ampacol \or \or \or \@addamp \or
\@acolampacol \or \@firstampfalse \@acol \fi
\edef\@preamble{\@preamble
\ifcase \@chnum
  \hfil$\relax\arraymode\@sharp$\hfil
  \or $\relax\arraymode\@sharp$\hfil
  \or \hfil$\relax\arraymode\@sharp$\fi}}
\def\@array[#1]#2{\setbox\@arstrutbox=\hbox{\vrule
  height\arraystretch \ht\strutbox
  depth\arraystretch \dp\strutbox
width\z@}\@mkpream{#2}\edef\@preamble{\halign \noexpand\@halignto
\bgroup \tabskip\z@ \@arstrut \@preamble \tabskip\z@ \cr}%
\let\@startpbox\@@startpbox \let\@endpbox\@@endpbox
\if #1t\vtop \else \if#1b\vbox \else \vcenter \fi\fi
\bgroup \let\par\relax
\let\@sharp##\let\protect\relax
\@arrayskip\@preamble}
\def\eqnarray{\stepcounter{equation}%
           \let\@currentlabel=\theequation
           \global\@eqnswtrue
           \global\@eqcnt\z@
           \tabskip\@centering              
           \let\\=\@eqncr
           $$%
         \halign to \displaywidth  \bgroup
          \eqnumphantom \@eqnsel
   \hskip\@centering                               
 $\displaystyle  \tabskip\z@ {##}$%
 &\global\@eqcnt\@ne \hskip 2\arraycolsep
      $ \displaystyle  \arraymode{##}$\hfil
 &\global\@eqcnt\tw@ \hskip 2\arraycolsep
      $\displaystyle\tabskip\z@{##}$\hfil
      \tabskip\@centering
 &{##}\tabskip\z@\cr}
\makeatother

\newcounter{mo}

\newcounter{AL}

\newcounter{ed}

\newcounter{bk}

\newcommand{\ti}[1]{\tilde{#1}}
\newcommand{\om}{\omega}

\newcommand{\de}{\delta}
\newcommand{\al}{\alpha}
\newcommand{\te}{\theta}

\newcommand{\be}{\beta}

\newcommand{\G}{\Gamma}

\newcommand{\ga}{\gamma}

\def\bea{\begin{eqnarray}\new\begin{array}{cc}}
\def\ee{\end{array}\end{eqnarray}}

\newcommand{\beq}[1]{\begin{equation}\label{#1}}
\newcommand{\eq}{\end{equation}}
\newcommand{\beqn}[1]{\begin{small} \begin{eqnarray}\label{#1}}
\newcommand{\eqn}{\end{eqnarray} \end{small}}
\newcommand{\p}{\partial}
\def\sq2{\sqrt{2}}
\newcommand{\di}{{\rm diag}}
\newcommand{\oh}{\frac{1}{2}}

\def\sl2{{\rm sl}(2, {\mathbb C})}

\newcommand{\Mp}{{\rm Mp}(2,\mZ)}

\def\f1#1{\frac{1}{#1}}

\def\mC{{\mathbb C}}
\def\mZ{{\mathbb Z}}

\def\frak{\mathfrak}

\def\gg{{\frak g}}

\def\gh{{\frak h}}

\def\bfc{{\bf c}}

\def\bfe{{\bf e}}
\def\bff{{\bf f}}

\def\bfs{{\bf s}}
\def\bft{{\bf t}}
\def\bfx{{\bf x}}

\def\clF{\mathcal{F}}
\def\clR{\mathcal{R}}

\def\clO{\mathcal{O}}

\def\clM{\mathcal{M}}

\def\clS{\mathcal{S}}

\def\bag2{{\bf g_2}}
\def\bas8{{\bf so(8)}}

\def\sr2{\sqrt{2}}
\newcommand{\ran}{\rangle}
\newcommand{\lan}{\langle}
\def\f1#1{\frac{1}{#1}}

\newtheorem{predl}{Proposition}[section]

\newtheorem{rem}{Remark}[section]

\newtheorem{lem}{Lemma}[section]

\begin{document}

 \begin{flushright}

 ITEP-TH-20/26\\
 IITP-18/26
 \end{flushright}
\vspace{3mm}

 \begin{center}
  {\Large M. Olshanetsky}\\
  \vspace{3mm}
{\bf\LARGE
 Epstein
vector zeta functions related to the ADE Lie algebras}\\
 \vspace{3mm}
\
  {\rm
National Research Centre "Kurchatov Institute",\\
  Academician Kurchatov square, 1, Moscow, 123182, Russia
}\\
   {\rm Institute for Information Transmission Problems RAS (Kharkevich Institute),
 \\  Bolshoy Karetny per. 19, Moscow, 127994,  Russia}\\
\vspace{2mm}
 {\footnotesize Email
 mikhail.olshanetsky@gmail.com
 }
\end{center} 

 \begin{abstract}
 We introduce a vector-valued generalization of the Epstein zeta functions associated with the root lattices of ADE-type Lie algebras. The quadratic forms defining these lattices correspond to the Gram matrices of the simple roots. Using the discriminant group D = P/Q, we construct vector-valued theta series that realize the Weil representation of the metaplectic group Mp(2,Z). The proposed Epstein vector zeta functions are obtained as the Mellin transform of these theta series. By exploiting the equivariance properties of the theta vectors, we derive a matrix functional equation of the Riemann type. We show that the existence of this functional equation is governed by a selection rule: it holds specifically for the subspace of C-invariant vectors, where C is the central element of Mp(2,Z). Finally, we provide a complete classification of the lattices and invariant subspaces for which this matrix functional equation is satisfied.

%

\end{abstract}

\section{Introduction.
Main results}
In his 1903 paper \cite{Ep}, Paul Epstein generalized the Riemann zeta function to multidimensional series associated with positive definite quadratic forms, defining what is now known as the Epstein zeta functions. He proved that these functions have  meromorphic continuations to the entire complex plane with  simple poles and derived a Riemann-type functional equation for it. The functional equation for the Epstein zeta function was later discussed in \cite{Te}.

The Epstein zeta functions found applications in different physical problems like
the Casimir effect, Bose-Einstein condensation, long-range interactions in quantum many-body systems
\cite{El,Ki}.

%


In this paper we consider the vector generalisation of the Epstein zeta functions attached to the vector
modular forms and the Weil representations. The closed construction was proposed by Stein \cite
{Ste}.

More specifically,
we consider the root lattices $Q$ of the ADE simple Lie algebras.
It should be noted that another class of zeta functions associated with root lattices is the Witten zeta functions \cite{KMH,Wi,Za}. This zeta function appears when calculating the partition function in the two-dimensional Yang--Mills theory. Later, to calculate the partition function in the two-dimensional Yang--Mills theory with a spontaneous broken gauge symmetry
 we modify the Witten zeta function \cite{LO}.

In our case we equip the root lattices  with even positive definite quadratic forms defined by the Gram matrices of the simple roots. Let $P$ is the weight lattice. This lattice is dual to the root lattice.
The quotient group $D=P/Q$ is the discriminant group. It is a finite abelian group of order
 $l=$ord$\,D$. This data allows one to define the Weil representation of the metaplectic group
$\Mp$. Let $\mathbb{H}\subset\mC$ be the upper half plane. The  metaplectic group acts in
space of holomorphic vector functions $\clO(\mathbb{H})^{l}$ bounded at infinity.
It is the Weil representations.
Let $S$ $T$  are the generators of inversion and the shift  and $C=S^2$ is the
central element.
\emph{The vector-valued theta-functions} (VVTF), constructed by this data,   belong to this space. The theta-vectors are $S$ and $T$ equivariant, but they are not necessarily $C$-invariant.
The $C$-invariant theta-vectors are modular forms of weight $k$, where $k=r/2$, and $r$ is the rank of of the lattice. The vector-valued modular forms
were appeared implicitly in the seminal paper of A.Weil \cite{We}.
This construction later appeared explicitly in the papers \cite{Bor} pages 164-173 and \cite{BG,Br,GN,Sch}.

The Mellin transform (MT) is well defined on subset of the theta vectors  that decay rapidly to infinity
(the cusp forms).
This subset is obtained from the set of the vector-valued theta-functions by subtracting the
fixed vector $e_0$. The Mellin  transform of this subset is proportional to \emph{the Epstein vector zeta functions} (EVZF). If the Epstein vector zeta functions are C-invariant, then they satisfy the matrix functional equation (\ref{fe}) (FE).
The FE follows from the S-equivariance of VVTE.
Thus, the C-invariance plays the role of selection rule for lattices for which the FE holds.

Schematically this construction looks as follows
\beq{01}
\begin{array}{c}

{\rm ADE~lattices}\to \,{\rm Weil~representation~of~}\Mp\to  \\
  \to{\rm\,VVTF\,}\to~ (C^{inv}\wedge{\rm cuspidality})\stackrel{\rm MT}\to {\rm EVZF}\to{\rm FE}
\end{array}
\eq


\section{Root lattices and Metaplectic group $\Mp$}

\subsection{Root and weight lattices}

Let $\gg$ be the simple simply-laced Lie algebra (the ADE types) and $R$ is the corresponding root system.
Consider the subsystem of simple roots $\Pi=\{\al_j\,|\,j=1,\ldots,r\}\subset R$, where $r$ is the rank of $\gg$. They form a basis in the Cartan subalgebra $\gh$.

The Gram matrix $\clR$ of the basis $\Pi$
\beq{gr}
\clR_{jk}=(\al_j,\al_k)
\eq
for the ADE algebras coincides with the Cartan matrix of the algebra $\gg$.

The fundamental weights $\Upsilon=\{\varpi_i\}$ is the system of vectors the dual to the systems of
simple roots $\Pi$.
\beq{du}
(\al_j,\varpi_a)=\de_{ja}\,.
\eq
Therefore, $G(\Upsilon)=G(\Pi)^{-1}$.

Let $Q$ be the root lattice
\beq{rt}
Q=\{\ga=\sum_{j=1}^rn_j\al_j\,,~~\al_j\in\Pi\,,~~n_j\in\mZ\}\,.
\eq
The the Gram matrix (\ref{gr}) defines the quadratic form $\clR(\ga)=(\ga,\ga)=\clR_{jk}n_jn_k$ on the lattice $Q$.
This form is  positive definite and even
\beq{el}
(\ga,\ga)\in 2\mZ\,.
\eq

 The weight lattice
$$
P=\{\be=\sum_{a=1}^rm_a\varpi_a\,,~~\varpi_a\in\Upsilon\,,~~m_a\in\mZ\}
$$
 is dual to the root lattice $Q^\vee=P$.
 The quotient $D(Q)=P/Q$ is called the discriminant group  of the root lattice. It is a finite cyclic group $\mZ_l=\mZ/l\mZ$ of order $l$ except
the algebra so(4n) ($D_{2n}$). In the latter case $P/Q\sim\mZ_2\times\mZ_2$. The order of $P/Q$ is
\beq{or}
l:={\rm ord}\,(P/Q)=\det\,\clR=|D(Q)|\,.
\eq

Define the subset of the fundamental weights $\ti \Upsilon\subset\Upsilon$ that
represent of non-identity cosets of P/QP/Q
P/Q."
\beq{bd}
\ti \Upsilon=
\{\bm{\varpi}=(\varpi_1,\ldots,\varpi_{l-1})\,,~\varpi_a\notin Q\,,~\varpi_a\neq\varpi_b~{\rm mod}~ Q\}\,.
\eq
These vectors along with $\varpi_0\sim Q$ form the induced basis in the quotient $D(Q)$
\beq{dq}
D(Q)\sim\varpi_0\cup\ti\Upsilon\,.
\eq

Let $\al_i$ be the simple roots dual to the weights $ \varpi_i\in\ti\Upsilon$ (\ref{bd}).
For the ADE root system the set of such $i$  takes the form
$$
\begin{tabular}{|c|c|c|}
    \hline
    \hline
    Root system &Set $i$& $D(Q)$\\
    \hline
    \hline
    $A_n$ & all $i$ & $\mZ_{n+1}$ \\
    \hline
       $D_n$ $n$ odd & $i=1,n-1,n$ & $\mZ_4$\\
     \hline
       $D_n$ $n$ even & $i=1,n-1,n$ & $\mZ_2\oplus\mZ_2$\\
    \hline
    $E_6$ & $i=1,6$ & $\mZ_3$\\
    \hline
    $E_7$ &$i=7$ & $\mZ_2$\\
    \hline
    $E_8$& no such $i$ & Id\\
    \hline
     \hline
  \end{tabular}
  $$
  \begin{center}
 \textbf{ Table 1} Non-trivial weights.
  \end{center}
 \textbf{ Remarks to Table 1.}\\
1. The numeration of roots  has been agreed with  \cite{Bo}.\\
2. For $D_n$ the weights $\varpi\in\ti\Upsilon$ define the vector, the left and the right
spinor representations.\\
3. There are the following  classes in $D(Q)$ for $E_6$
$$
[\varpi_2]=[\varpi_4]\in Q\,,~~[\varpi_1]=[\varpi_5]\,~~[\varpi_3]=[\varpi_6]\,,
$$
\beq{e6}
2[\varpi_1]=[\varpi_3]\,.
\eq


\subsection{Metaplectic group }

\paragraph{SL$(2,\mZ)$.}
Consider the
 group of the unimodular second order matrices
$$
{\rm SL}(2,\mZ)=\left\{g=\left(
                         \begin{array}{cc}
                           a & b \\
                           c & d \\
                         \end{array}
                       \right)\,,~~a,b,c,d\in\mZ\,,~~ad-bc=1\right\}\,.
$$
Let $\mathbb{H}=\{\tau\in\mC\,|\,\Im m\tau>0\}$ be the upper half plane.
The group SL$(2,\mZ)$ acts on $\tau$
 by the M\"{o}bius transformation
$$
\tau\to g\cdot\tau=\frac{a\tau+ b}{ c\tau+d}\,.
$$
In fact, this action is defined for the projective group ${\rm PSL}(2,\mZ)={\rm SL}(2,\mZ)/\pm Id_2$, where\\
$\pm Id_2$ is the center of ${\rm SL}(2,\mZ)$.

The group ${\rm SL}(2,\mZ)$ is generated by two matrices
\beq{st}
\bfs=\left(
    \begin{array}{cc}
      0 & -1 \\
      1 & 0 \\
    \end{array}
  \right)
  \,,~~\bft=\left(
           \begin{array}{cc}
             1 & 1 \\
             0 & 1 \\
           \end{array}
         \right)\,.
\eq
It means that $s$ and $t$  acts as $\tau\to-1/\tau$ and $\tau\to\tau+1$.
In addition, introduce the central element $\bfc=\bfs^2=-Id_2$.


 Define the representations of SL$(2,\mZ)$ in the space  $\clS^0$.
 The functions $f(\tau)$ from this space  satisfy the following conditions.
 Let $\tau=x+\imath y$. Then
 \beq{fd}
 \clS^0\{f(\tau)\,,~\tau\in \mathbb{H}\}\,,~~\left\{
 \begin{array}{ll}
  1.\, \lim_{y\to 0^+}y^n\left|\frac{\p^{a+b}f(\tau)}{\p^a_x\p^by}\right|=0 & \forall\, n\in\mathbb{N} \\
  2.\, {\rm sup}_{\tau\in \mathbb{H}}|\tau^my^{-n}\p^a_x\p_y^bf(\tau)|<\infty& \forall\, m,n,a,b\in\mathbb{N} \,.
 \end{array}
 \right.
\eq
The representation $\pi_k$ of ${\rm SL}(2,\mZ)$ for $k=0,1,2,\dots$ has the form
\beq{rep}
\pi_k(g) f(\tau)=\al(g,\tau)^{-k} f(g\cdot\tau)\,,~~\al(g,\tau)=c\tau+d\,.
\eq
The analytic prerfactor $\al(g,\tau)$ satisfies
the one-cocycle relation
\beq{cr}
\al(g_1g_2,\tau)=\al(g_1,g_2\cdot\tau)\al(g_2,\tau)\,.
\eq
For $k $ even the representation (\ref{rep}) is lifted from the linear representation of ${\rm PSL}(2,\mZ)$.


\paragraph{Metaplectic group $\Mp$.}

We need to define the action for $k$ being half-integer.
 To this end, consider a
 double cover of   SL$(2,\mZ)$.
It is the metaplectic group Mp$_2(\mZ)$.
This group is  the extension of  SL$(2,\mZ)$ by the group $\mZ_2=\{\pm 1\}$
  $$
  1\to \mZ_2\to {\rm Mp}(2,\mZ)\to {\rm SL}(2,\mZ)\to 1\,,
  $$
  The cohomology group responsible for this extension is $H^2({\rm SL}(2,\mZ),\mZ_2)$.

  The group  Mp$_2(\mZ)$ has  three generators $\ti\bfs:=S$,
 $\ti\bft:=T$ covering
$\bfs$ and $\bft$ (\ref{st}) and the center element  $\ti \bfc=C$
$({\rm Mp}_2(\mZ)=\{S\,,\,T\,,\,C\})$.
  They
satisfy the Mp$(2,\mZ)$ relations
\beq{m2}
1.S^2=C\,,~~2.(ST)^3=C\,,~~3.C^2=Id\,.
\eq



\subsection{Weil representations}

 The Weil representations of $\Mp$ is   the generalizations of the representations (\ref{rep}) of SL$(2,\mZ)$.
 Consider the vector-space
  \beq{fh}
  \clS^0_l=\{\phi\in\clO(\mathbb{H})^{l}\,|\,\phi\in\clS^0~(\ref{fd})\}\,.
  \eq

  Let $\alpha_k(g,\tau)$ be the analytic prefactor (\ref{rep}) and $\rho_l(g)$
  are the matrices of order $l$ (the Weil matrices).
  The Weil representation in the space $\clS^0_l$ is defined as
  \beq{wr7}
  \pi_k( g)f(\tau)=\alpha^{-1}_k(g,\tau)\rho_l(g)f(g\tau)\,,~~g\in\Mp\,,~~f\in\clS^0_l\,.
  \eq

 Define the distribution $(\clS^0)^\vee_l$ over $\clS^0_l$
 \beq{wr10}
 (\clS^0)^\vee_l=\{f\,|\,\lan f,\phi\ran<\infty ~\forall\,\phi\in\clS^0_l\,.
  \eq
  Here the brackets $\lan f,\phi\ran$ mean taking the integral over the upper half plane $\mathbb{H}$ and the trace over $\mC^l$.

Consider in details the analytic prefactors and the Weil matrices.

\paragraph{Analytic prefactors}
They are the functions on the upper half-plane (see (\ref{rep})) and take the form
\beq{alp}
\al_k(g,\tau)=(c\tau+d)^{k}\,,~~k\in\oh\mZ^+\,.
\eq
Due to (\ref{cr}) it is the one-cocycle and $[\al]\in H^1({\rm SL}(2,\mZ),\mC^*)$.

For $\Mp$ the relation (\ref{cr}) is replaced on
\beq{cr1}
\al_k(g_1g_2,\tau)=\al_k(g_1,g_2\cdot\tau)\al_k(g_2,\tau)\om_k(g_1,g_2)\,,
\eq
where $\om_k(g_1,g_2)$ is the two-cocycle $[\om_k(g_1,g_2)]\in H^2({\rm SL}(2,\mZ),\mZ_2)$.


 On the generators $(\bfs,\bft,\bfc)\in$SL$(2,\mZ))$ the prefactors $\al_k$ (\ref{alp}) take the form
 \begin{subequations}\label{api}
  \begin{align}
 &\al_k(\bfs|\tau)=\tau^{k}\,,\\
 & \al_k(\bft|\tau)=1\,,\\
 &\al_k(\bfc|\tau)=e^{\pi\imath k}=\bfe(k)\,.
 \end{align}
  \end{subequations}

  Since $\al_k(\bfs|\bfs\tau)=e^{\pi\imath k}\tau^{-k}$, we have $\al_k(\bfs^2|\tau)=e^{\pi\imath k}$.
  Thus, $\al_k(\bfs^2|\tau)=\al_k(\bfc|\tau)$. Because $\rho_k(\bfs^2)=\rho_k(\bfc)$, the first relation in (\ref{m2}) is fulfilled. The last relation in (\ref{m2})  is provided by the nontrivial cocycle $\om_k(\bfc,\bfc)$.
  As a result,
the two-cocycle takes the form
 \begin{subequations}\label{om}
  \begin{align}
 &\om_k(\bfs,\bfs)=1\,,\\
 & \om_k(\bft,g)=1\,,~\forall g\\
 &\om_k(\bfc,\bfc)=e^{-2\pi\imath k}\,.
 \end{align}
  \end{subequations}
  Since  $\dim\,H^2({\rm SL}(2,\mZ),\mZ_2)=1$ the cocycle $\om_k(\bfc,\bfc)$ represents
  the  only one nontrivial class in $H^2({\rm SL}(2,\mZ),\mZ_2)$.
Note that due to (\ref{api}c) and (\ref{om}c) $\al_k(\bfc^2)=e^{4\pi\imath k}=1$.


 \paragraph{Weil matrices $\rho_l(\bfs)$ and  $\rho_l(\bft)$.}
Let $l=|D(Q)|$ is the determinant of Gram matrix (see Table 1).  In the space $\mC^l$ fix the
nontrivial fundamental weights basis
  \beq{cl}
  \mC^{l}=\{x^0\varpi_0+\sum_{a=1}^{l-1}x^a\varpi_a\,,~~\varpi_a\in\ti\Upsilon\,,~(\ref{bd})\}\,.
  \eq

Define three Weil matrices $\rho_l$ acting in $\mC^l$ corresponding to the three generators of $\Mp$
\beq{rks}
S\to\rho_l(\bfs)\,,~~\rho_l(\bfs)=
\f1{\sqrt{|D(Q)|}}\Phi\,,~
(\Phi)_{ab}=\bfe(-2(\varpi_a,\varpi_b))\,,~~(\bfe(x)=e^{\pi\imath x})
\footnote{We use here a non-standard notation for the exponent $\bfe$.}\,.
\eq
In particular, for $l=2$
\beq{r2}
\rho_2(\bfs)=\f1{\sqrt{2}}\left(
               \begin{array}{cc}
                  1& 1 \\
                 1 & -1 \\
               \end{array}
             \right)\,.
\eq
 The generator $T$ is represented by the diagonal matrix
\beq{TT}
T\to\rho_l(\bft)\,,~~
(\rho_l(t))_{a,a}=\bfe(-(\varpi_a,\varpi_a))\,.
\eq
\noindent
\paragraph{Weil matrices $\rho_l(\bfc)$.}
Because $C=S^2$, the Weil matrix representing the center  $\rho_l(\bfc)$  is defined as
\beq{CC}
\rho_l(\bfc)=\rho^2_l(\bfs)=\f1{D(Q)}\Phi^2\,.
\eq
It follows from (\ref{rks}) that
\beq{per}
\rho_l(\bfc)=\f1{|D(Q)|}(\Phi^2)_{ab}=
\f1{|D(Q)|}\sum_{c\in\ti\Upsilon}\bfe(-2((\varpi_a+\varpi_b,\varpi_c)))=\de_{(\varpi_a+\varpi_b,Q)}\,.
\eq
Thus
\beq{rok}
\rho_l(\bfc)=\left(\begin{array}{cc}
         1 & 0 \\
         0 & {\rm Inv}_l
       \end{array}\right)\,,~~({\rm Inv}_l)_{ab}=\de_{(\varpi_a+\varpi_b,Q)}\,.
\eq
 For $l=2$ (\ref{r2}) implies
\beq{rc2}
\rho_2(\bfc)=Id_2\,.
\eq
\paragraph{ l=2: A$_1$, E$_7$}.\\
In the cases $A_1$ and E$_7$
the discriminant group $D(Q)$ is isomorphic to $\mZ_2$.
The set $\ti\Upsilon$ has only one weight
$\varpi_1$. Therefore, $(\varpi_1+\varpi_1)\in Q$.
 It follows from (\ref{per}) that
\beq{roh}
\rho_{2}(\bfc)=Id_2\,.
\eq

\paragraph{A$_n$}.\\
It follows from Table  that for this root system $l=n+1$, $D(Q)\sim\mZ_{n+1}$.
 The fundamental weights $\Upsilon$ dual to the system of simple roots
 $$
 \Pi=\{\al_j=(\underbrace{0,\ldots,0}_{j-1},1,-1,\underbrace{0,\ldots,0}_{n-j}),~j=1,\ldots,n\}
 $$
 are
 \beq{fw}
 \Upsilon=\{\varpi_a=\f1{n+1}(\underbrace{n-a+1,\ldots,n-a+1}_{{a}} \underbrace{-a,\ldots,-a}_{{n+1-a}})\,,
 ~a=1,\ldots,n\}\,.
 \eq
 In this case $\ti\Upsilon=\Upsilon$.
 Thereby, the matrix $\rho_l(\bfc)$ has the form (\ref{rok}) and the order $n+1$.

 \paragraph{D$_{2m}$}.\\
For $D_{2m}$ $D(Q)\sim\mZ_2\oplus\mZ_2$ and there are two fundamental weights in $\ti\Upsilon$
corresponding to the left and right spinor representations. From (\ref{rc2}) we have
\beq{roh1}
\rho_{4}(\bfc)=Id_2\oplus Id_2\,.
\eq

In the rest cases the discriminant groups $D(Q)$ are isomorphic to $\mZ_l$, $l>2$.
It is possible to choose a weight $\varpi_1$ such the  weights in $\ti\Upsilon$
take the form $\varpi_a=a\varpi_1$, $a=1,\ldots,l$.
It means that in (\ref{rok})
\beq{rok1}
({\rm Inv}_l)_{ab}=\de_{a+b,0\,({\rm mod}\,l)}\,.
\eq

\paragraph{D$_{2m+1}$}.\\
In this case $D(Q)\sim\mZ_4$, $l=4$.
Choose weights in $\varpi_a$ in $\ti\Upsilon$ as follows.
The Cartan subalgebra of so$(4m+2)$ is  the space $\mC^{2m+1}$. The weights in the set $\ti\Upsilon$ are
identified with the vectors
\beq{sow}
 \varpi_1=\oh(1,1,\ldots,-1)\,,~~
\varpi_2=(1,0,\ldots,0)\,,~~\varpi_3=\oh(1,1,\ldots,1)\,.
\eq
 Here $\varpi_1$ and $\varpi_3$
are the weights of spinor representations with opposite chirality, and  $\varpi_2$  is basic weight of the vector representation.
Then $\varpi_k$=$k\varpi_1$ (mod$\,Q$), $k=0,1,2,3$.
The Gram matrix of this basis is
$$
\rho_4(\bfs)=\oh\left(
              \begin{array}{cccc}
                1 & 1 & 1 & 1 \\
                1 & \imath(-1)^m & -1 & -\imath(-1)^m \\
                1 & -1& 1 & -1 \\
                1 & -\imath(-1)^m & -1 &\imath(-1)^m  \\
              \end{array}
            \right)
$$
 It can be checked that
the matrix $\rho_4(\bfc)=\rho^2_4(\bfs)$ has the form (\ref{rok}).
\paragraph{E$_6$}.\\
For $E_6$ the quotient $P/Q=D(Q)\sim\mZ_3$.
According with Remarks to Table 1 the set $\ti\Upsilon$ contains two weights  $\varpi_1$ and $\varpi_6$.
 The Gram matrix of this basis is
$$
(\varpi_1,\varpi_1)=(\varpi_6\varpi_6)=\frac{4}3\,,~~(\varpi_6,\varpi_1)=\frac{2}3\,.
$$
 It implies
 $$
\rho_3(\bfs)=\f1{\sqrt{3}}\left(
              \begin{array}{ccc}
                1 & 1 & 1  \\
                1 & \om & \om^2   \\
                1 & \om^2 & \om  \\
              \end{array}
            \right)\,,~~~\om=\exp\frac{2\pi\imath}3\,.
$$

  The matrix $\rho_3(\bfc)$ has the form (\ref{rok})
  and the order $3$.
  \paragraph{E$_8$}.\\
  For E$_8$ $P=Q$ and therefore $\rho_1(\bfc)=1$.


  \paragraph{Weil representation}.\\
  Let's derive the specific formulas for the  generators $S$, $T$ and $C$
  of  the Weil representation  (\ref{wr7}).
   Using (\ref{api}) we find
  \begin{subequations}\label{re5}
  \begin{align}
 &  \pi_k(S)F(\tau)=\tau^{-k}\rho_l(\bfs)F(-1/\tau)\,,\\
 &    \pi_k(C)F(\tau)=e^{-\pi ik}\rho_l(\bfc)F(\tau)\,,\\
 &\pi_k(T)F(\tau)=\rho_l(\bft)F(\tau+1)\,.
 \end{align}
  \end{subequations}



\subsubsection{C-invariance}

Let $\bfx_l=(x_0,x_1,\ldots,x_{l-1})$ be a vector in $\mC^{l}$ (\ref{cl}).
Its C-invariance means that
\beq{ci}
\rho_l(\bfc)\bfx=e^{\pi ik}\bfx\,,~~k=r/2\,.
\eq
Let $\mC^l_C$ be the subspace of $\mC^l$ satisfying this condition
\beq{clc}
\mC^l_C=\{\bfx\,|\,\rho_l(\bfc)\bfx=e^{\pi ik}\bfx\}\,.
\eq

\vspace{-0.3cm}
\paragraph{Algebras that do not permit C-invariant vectors}.\\
First prove that  in the following four cases the invariant subspace is trivial $\mC^l_C=\{\bfx=0\}$.
\vspace{-0.7cm}
\subparagraph{A$_1$ and E$_7$}.\\
In this case $l=2$, $k=1$ for SL(2) and $k=7/2$ for E$_7$.
It follows from (\ref{roh}) that $\rho_2(\bfc)=Id_2$.
 The invariance condition (\ref{ci}) means that
 $\bfx=-\bfx$ for  SL(2) and $\bfx=-\imath \bfx$ for E$_7$. Thus,
 \beq{se}
\mC^2_C({\rm SL}(2))=\{\bfx=0\}\,,~~\mC^2_C({\rm E}_7)=\{\bfx=0\}\,
\eq
\vspace{-0.5cm}
\subparagraph{ A$_{2m+1}$}.\\
In this case  $k=m+\oh$, $l=2m+2$ .It follows from (\ref{rok}) and (\ref{ci}) that
$$
\rho_l(\bfc)\bfx_{2m+2}=e^{\pi im}\imath\bfx_{2m+2}
$$
Since $\rho^2_l(\bfc)=Id_l$, we have $-\bfx_{2m+2}=\bfx_{2m+2}$
 and therefore
\beq{sloo}
\mC^{2m+2}_C({\rm SL}(2m+2))=\{\bfx_{2m+2}=0\}\,.
\eq
\vspace{-0.5cm}
\subparagraph{D$_{2m+1}$}.\\
Now $l=4$ and  $k=m+1/2$. This case is similar to the previous one
because the phase factor $e^{\pi\imath k}$ is proportional to $\imath$.
Therefore,
\beq{so1}
\mC^4_C({\rm D}_{2m+1})=\{\bfx=(0,0,0,0)\}\,.
\eq

\paragraph{Algebras that have the C-invariant vectors}.\\
We list the nontrivial cases.
\vspace{-0.8cm}
\paragraph{A$_n$}.\\
$l=n+1$, $k=n/2$.\\
It follows from (\ref{ci})) that the C-invariance means that
$$
(x_0,x_n,\ldots,x_{1})=e^{\pi\imath n/2}(x_0,x_1,\ldots,x_{n})\,.
$$
 $\bf{n=4m}$, $k=2m$, $l=4m+1$, $\bfx=(x_0,x_1,\ldots,x_{4m})$. \\
  The invariant subspace has the form
\beq{sle}
\mC^{4m+1}_C({\rm A}_{4m+1})=\{\bfx=(x_0,x_1,\ldots,x_{2m},x_{2m},\dots,x_1)\}\,.
\eq
 $\bf{n=4m+2}$,  $k=2m+1$, $l=4m+3$, $\bfx=(x_0,x_1,\ldots,x_{4m+2})$. \\
\beq{slo}
\mC^{4m+3}({\rm A}_{4m+2}))=\{\bfx_{4m+3}=(0,x_1,\ldots,x_{2m+1},-x_{2m+1},\dots,-x_1)\}\,.
\eq
 \vspace{-0.6cm}

\paragraph{D$_{2m}$}.\\
$l=4$, $k=m$. It follows from (\ref{roh1})
$$
\rho_4(\bfc)(x_0,x_1,x_{2},x_{3})^T=(x_0,x_{1},x_2,x_3)^T\,.
$$
Thus, the C-invariance does not impose any constraints for this algebra
\beq{sof}
\mC^4_C({\rm D}_{2m})=\mC^4\,.
\eq
\vspace{-0.8cm}
\paragraph{E$_6$}.\\
$l=3$, $k=3$, $\mC^3(E_6)=\{\bfx_3=x_0,x_1,x_6\}$,
$\rho(\bfc)\bfx=-(x_0,x_6,x_1)$. Thus, the C-invariant subspace has the form
\beq{s6}
\mC^3(E_6)=\{0,x_1,-x_1)\}\,.
\eq
\vspace{-0.8cm}
\paragraph{E$_8$}.\\
$l=1$, $k=4$. The condition (\ref{ci}) is satisfied for  $\mC$.




\section{Theta vectors}

\subsection{Definition and modular properties}

By means of the quadratic form $\clR$ (\ref{gr})
 define $l=ord D(Q)$ theta series  attributed to the set of fundamental weights $\ti\Upsilon$ (\ref{bd})
\beq{tc}
\te_a(\tau)=\sum_{\ga\in Q+\varpi_a}\bfe(\tau(\ga,\ga))\,,~~
\varpi_0=0\,,~-\varpi_a\in\ti\Upsilon\,,\,.
\eq
For $a\neq 0$ $\te_a(\tau)\in (\clS^0)^\vee_l$ (\ref{wr10}).
Rewrite these series as
\beq{tc1}
\te_a(\tau)=\sum_{\ga\in Q}\bfe(\tau\clR_a(\ga))\,,~~(a=0,1,\ldots,l)\,
\eq
where
\beq{ra}
\clR_a(\ga)=(\ga+\varpi_a,\ga+\varpi_a)=(\ga,\ga)+2(\ga,\varpi_a)+(\varpi_a,\varpi_a)\,.
\eq
Applying the Poisson summation formula we obtain
\beq{sto}
\te_a(-1/\tau)=
\frac{(-\imath\tau)^{k}}{l^{1/2}}
\sum_{\varpi_b\in\bm{\varpi}}
\te_b(\tau)\bfe(2(\varpi_a,\varpi_b))\,,~~(k=r/2\,,~\bfe(x)=e^{\pi\imath x})\,.
\eq

Define  the column theta vector
\beq{tc2}
\bm{\Theta}(\tau)=(\te_0(\tau),\te_1(\tau),\ldots,\te_{l-1}(\tau))^T\,.
\eq
It can be considered as an element of $(\clS^0)^\vee_l$ (\ref{wr10}).
The Weil representation (\ref{wr7}) on the theta vectors take the form
\beq{wt}
  \pi_k(\tilde{g})\bm{\Theta}(\tau)=\alpha^{-1}_k(g,\tau)\rho_l(g)\bm{\Theta}(g\tau)\,.
  \eq
From  (\ref{sto}) we obtain
\beq{st11}
\bm{\Theta}(-1/\tau)=(-\imath\tau)^{k}\rho^{-1}_l(\bfs)\bm{\Theta}_l(\tau)\,.
\eq
It follows from (\ref{re5}a) that the vector $\bm{\Theta}(\tau)$ is equivariant with respect to the $S$ action.

Consider the action of the $T$-generator.
It follows from (\ref{tc1}) that
$$
\te_a(\tau+1)=\sum_{\ga\in Q}\bfe(\tau\clR_a(\ga))\bfe(\clR_a\emph{}(\ga))\,.
$$
Because the root lattice is even (\ref{el}) and the fundamental weights are dual to the
simple roots the first two terms in (\ref{ra}) are even integers. Thereby
$$
\te_a(\tau+1)=\te_a(\tau)\bfe(\varpi_a,\varpi_a)\,.
$$
Then
 \beq{teq}
 \bm{\Theta}_l(\tau+1)=\rho^{-1}_l(\bft)\bm{\Theta}_l(\tau)\,,
\eq
where $\rho_l(\bft)$ is (\ref{TT}). Thus, $\bm{\Theta}_l(\tau)$ is equivariant with respect to
the $T$ action.


\subsection{C-invariant theta vectors}

It follows from the above that the C-invariant theta vectors are equivariant with respect
to the action of the  entire group $\Mp$. Thus, the C-invariant theta vectors are the modular forms.

\paragraph{List of  the C-invariant theta vectors}

We obtain the list using the classification (\ref{sle})--(\ref{s6}) and E$_8$.
\vspace{-0.5cm}
\paragraph{ A$_{n}$, n=4m}.\\
\beq{sl1}
\bm{\Theta}_{A_{4m}}(\tau)=(\te_0,\te_1,\ldots,\te_{2m},\te_{2m}\ldots\te_1)\,.
\eq
\vspace{-0.8cm}
\paragraph{A$_{n}$,  n=4m+2}.\\
\beq{sl2}
\bm{\Theta}_{A_{4m+2}}(\tau)=(0,\te_1,\ldots,\te_{2m+1},-\te_{2m+1},\ldots,-\te_1)\,.
\eq
\vspace{-0.8cm}
\paragraph{D$_{2m}$}.\\
\beq{so121}
\bm{\Theta}_{D_{2m}}(\tau)=(\te_0,\te_1,\te_{2},\te_{3})\,.
\eq
\vspace{-0.8cm}
\paragraph{E$_6$}.\\
\beq{se60}
\bm{\Theta}_{{\rm E}_6}=(0,\te_1,-\te_1)\,.
\eq
\vspace{-0.8cm}
\paragraph{E$_8$}.\\
\beq{se66}
\bm{\Theta}_{{\rm E}_8}=\te_0\,.
\eq


\subsection{Modular and cusp forms }

Define the space $\clF_l$ of holomorphic maps of the upper half plane to $\mC^l$
\beq{mup}
\clF_l=\{\bff\in\clO(\mathbb{H},\mC^l\}\,.
\eq
In the basis (\ref{cl})
$$
\bff=\sum_{a=0}^{l-1}f_a(\tau)\varpi_a
$$
the components $f_a(\tau)$ are holomorphic.

\paragraph{Modular forms}
\emph{The vector valued \textbf{modular form}  $\bm{\Phi}_k(\tau)$  of the weight $k$ is the vector
in the space $\clF_l$
equivariant with respect to the  $\Mp$ action} (\ref{wr7})
\beq{mf}
\bm{\Phi}_k(\ti g\cdot\tau)=\pi^{-1}_k(\ti g)\bm{\Phi}_k(\tau)\,,~~
(\bm{\Phi}_k(g\cdot\tau)=\al_k(g,\tau)\rho^{-1}_k(g)\bm{\Phi}_k(\tau))\,.
\eq
Denote the space of modular forms $M_{k,l}$.

To be more specific the modular forms  of the weight $k$ satisfies the equations
\begin{subequations}\label{re7}
  \begin{align}
 & F(-1/\tau)=\tau^{k}\rho^{-1}_l(\bfs)F(\tau)\,,\\
 &F(\tau+1)=\rho^{-1}_l(\bft)F(\tau)\,  ,\\
 & F(\tau)=e^{\pi ik}\rho^{-1}_l(\bfc)F(\tau)\,.
 \end{align}
  \end{subequations}
The last equation means that the modular forms should be $C$-invariant.
It follows from (\ref{st11}) and (\ref{teq})  the C-invariant theta vector are modular forms
from the space $M_{k,l}$.

\paragraph{Cusp forms.}
The cusp forms are the  forms that vanish at the cusp $y\to\infty$ ($y=\Im m\,\tau$).
They belong to the subspace
the space $\clF_l$ of holomorphic maps of the upper half plane to $\mC^l$
\beq{mup0}
\clF^0_l=\{\bff\in\clF_l\,(\ref{mup})\,|\,\lim_{y\to\infty}\bff=0\}\,.
\eq

The cusp is invariant under the  action of the subgroup $\Gamma$ generated by $T$.
Denote the space of the cusp forms $C_{k,l}\subset M_{k,l}$
\beq{ck}
C_{k,l}=\{\bm{\Phi}\in M_{k,l}\,|\,\bm{\Phi}_k|_{\infty}=0\}\,.
\eq
Define the map $ev_\infty\,:\,\bm{\Phi}_k\to\bm{\Phi}_k|_{\infty}$.
This map can be included in the exact sequence
$$
0\rightarrow C_{k,l}\rightarrow  M_{k,l}\xrightarrow{ev_\infty}V_\infty\rightarrow 0
$$

The theta series are the cusp forms if the constant term vanishes.
It follows from (\ref{tc}) that $\te_a$ for $a\neq 0$ are the cusp forms,
 while $\te_0$ is not a cusp form. Thereby the vector
\beq{mp1}
\bm{\Psi}(\tau)=\bm{\Theta}(\tau)-e_0\,,~~e_0=(1,0,\ldots,0)^T\
\eq
looks like a cusp form vector. But the vector $e_0$ is not the modular form and
the substraction of $e_0$ breaks the equivariance of $\bm{\Psi}(\xi)$
\beq{mp5}
\bm{\Psi}(\tau)=\pi_k(g)\bm{\Psi}(\tau)+c(g)\,,~~c(g)=(\pi_k-Id_l)e_0\,.
\eq
Here $c(g)$ is the one-cocycle $[c(g)]\in H^1(\Mp,V)$, where $V$ is the $\Mp$ module.
If $V=\mC^l$ or $\clF_l$ this cocycle is trivial. But for $\clF_l^0$ the cocycle is non-trivial.

\begin{rem}
 Note that the theta vectors live in  $(\clS^0)^\vee_l$ (\ref{wr10}),
  being modular forms, are recovered as the intersection of
 $(\clS^0)^\vee_l$ with $\clF_l$.
\end{rem}


\section{Epstein zeta functions}

\subsection{Definition}

The definition of Epstein zeta function on  the multi-dimensional lattices
 was given by Epstein himself \cite{Ep}.
 Here we specialise the lattices to be the root ADE lattices and the quadratic forms
 are the Gram matrices (\ref{gr}).

By means the shifted quadratic forms (\ref{ra})
define $l$ Epstein zeta functions
\beq{ez}
\zeta_a(s)=
\left\{
\begin{array}{cc}
  \sum_{\ga\in Q}\f1{(\clR_a(\ga))^{s}} & a\neq 0\\
  \sum_{\ga\in Q\setminus 0}\f1{(\clR(\ga))^{s}} & a=0
\end{array}
\right.
\eq
Our main interest the Epstein vector zeta function
$$
\bm{\zeta}_l(s)=(\zeta_0(s),\ldots,\zeta_{l-1}(s)))^T\,.
$$
For convenience normalise it as
\beq{mz1}
\bm{\Xi}_l(s)=\frac{\G(s)}{\pi^s}\bm{\zeta}(s) \,.
\eq

The purpose of this section is
 to establish the matrix functional equation that the vectors $\bm{\Xi}_l(s)$ satisfy
 under special conditions.

\vspace{-0.5cm}
\paragraph{Representations of $\Mp$ related to the rotated the complex plane}.\\
Let $\xi=-\imath\tau$. Since $\tau\in\mathbb{H}$, $\xi\in\mC^+=\{\xi\in\mC\,|\,\Re e\xi>0\}$.
The passage from $\tau$ to $\xi$ corresponds to the conjugation of the matrix
$g(a,b,c,d)\in$SL$(2,\mZ)$ by the diagonal matrix
$\di(e^{\pi\imath/4},e^{-\pi\imath/4})$. It implies the replacement $g(a,\imath b,-\imath c,d)$.

 The group SL$(2,\mZ)$ acts on the plane  $\Re e\xi>0$ as
$$
\bfs\xi=1/\xi\,,~~\bft\xi=\xi-\imath\,,~~\bfc\xi=\xi\,.
$$

Let $\clS'$ be the space obtained from the space $\clS^0$ (\ref{fd}) after the rotation
\beq{ss}
\clS^0\xrightarrow{\tau\to \xi}\clS'\,,~~(\xi=y-\imath x)
\eq
 and $\clS'_l$ is its vector generalisation (\ref{fh}).
   The representation of the generator $S$ in $\clS'_l$ takes the form
$$
\al_k(\bfs|\xi)=e^{\pi\imath k/2}\xi^k\,,~~
$$
\beq{sac}
\pi_k(S)f(\xi)=A(S|\xi)\rho_l(\bfs)f(1/\xi)\,,~~ A(S|\xi)=e^{-\pi\imath k/2}\xi^{-k}
\eq
For the central element $C$ we have
\beq{cac}
\al_k(\bfc|\xi)=e^{\pi\imath k}\,,~~
\pi_k(C)f(\xi)=A(C)\rho_l(\bfc)f(\xi)\,, ~~A(C)=e^{-\pi\imath k}
\eq
The both actions $\pi_k(S)$ and $\pi_k(C)$ preserve $\clS_l'$.

Let $(\clS_l')^\vee$ be the space of distributions on $\clS_l'$. Then the
the dual operators $\pi\vee_k(S)$ and $\pi^\vee_k(C)$ act on $(\clS_l')^\vee$.

\vspace{-0.5cm}
\paragraph{The theta series on the right half-plane $\mC^+$}.\\
Rewrite the theta function (\ref{tc1}) in terms of the variable $\xi:\,$   $\te_a(\tau)\stackrel{\tau=\imath\xi}{\to}\te_a(\xi)$
\beq{th}
\te_a(\xi)=\sum_{\ga\in Q}e^{-\pi\xi\clR_a(\ga)}\,,~~
\bm{\Theta}(\xi)=(\te_0(\xi),\te_1(\xi),\ldots,\te_{l-1}(\xi))^T\,.
\eq
The equivariance of $\bm{\Theta}(\xi)$
with respect to the $S$ action takes the form
\beq{st1}
\bm{\Theta}(1/\xi)=e^{\pi\imath k/2}\xi^{k}\rho^{-1}_l\bm{\Theta}(\xi)\,.
\eq
The C-invariance means that
\beq{ci1}
\bm{\Theta}(\xi)=e^{\pi\imath k}\rho^{-1}_l(\bfc)\bm{\Theta}(\xi)\,.
\eq

\subsection{Mellin transform}


The Mellin transform is defined for $f\in\clS'$ (\ref{ss})
\beq{me1}
\clM f(s)=\hat f(s)=\int_0^\infty\xi^{s}f(\xi)\frac{d\xi}\xi\,.
\eq
Let $\widetilde{\clS'}=\clM(\clS')$ be the image of $\clS'$. It is the space of
entire functions with very rapid decay on vertical lines.
The Mellin image of the distributions $\widetilde{\clS'}^\vee=\clM(\clS')^\vee$ is the space of meromorphic functions with polynomial growth on vertical lines (see details in \cite{IK}).
In particular, the Mellin transform of the distribution $\xi^{-k}$ takes the form
\beq{ded}
\clM(\xi^{-k})(s)=\de(s-k)\,.
\eq



\paragraph{Mellin intertwiners}.\\

Consider the subgroup  $\mZ_4\subset\Mp$ generated by $S$ and $C$.
It follows from (\ref{sac}), (\ref{cac}) and (\ref{ded}) that
Let $\bff\in(\clS_l')^\vee$.
\beq{me2}
\clM(\pi^\vee_k(S)\bff(\xi))=\widetilde{\pi^\vee_k}(S) \widetilde{\bff}(s)=e^{-\pi\imath k/2}\rho_l(\bfs) \widetilde{\bff}(s-k)\,,~~\widetilde{\bff}\in\widetilde{\clS'}^\vee
\eq
\beq{me3}
\clM(\pi^\vee_k(C)\bff(\xi))=\widetilde{\pi^\vee_k}(C) \widetilde{\bff}(s)=
e^{-\pi\imath k}\rho_l(\bfc)\widetilde{\bff}(s)\,.
\eq
Thus, the Mellin transform intertwines representations
$\widetilde{\pi^\vee_k}\clM=\clM\pi^\vee_k$.


\paragraph{C-invariance}
Define the subspace $(\clS_l'^C)^\vee$ of the $C$-invariant vectors in $(\clS_l')^\vee$
 \beq{e0}
(\clS_l'^C)^\vee=\{\bff\in(\clS_l')^\vee\,|\,\bff=\pi^\vee_k(C)\bff\}\,.
\eq
or
$$
\bff(\xi)=e^{-\pi\imath k}\rho_l(\bfc)\bff(\xi)\,.
$$
 Since the action of the operator $\pi_k(C)$
on $\bff(\xi)$ coincides with the action of  $\pi_k(C)$
on $F(\tau)$ (\ref{re7}b) the form the subspace $(\clS_l'^C)^\vee$ is defined by the invariant space
$\mC^l_C$ (\ref{clc}) described in {\bf 2.4.1} (\ref{se})--(\ref{s6}).

\bigskip
\noindent
\paragraph{Mellin transformation of vector theta series.}
The Mellin transform of the exponents in the series (\ref{th}) is defined as
\beq{mt13}
\clM(\exp -\pi\xi \clR_a(\ga))(s)=\int_0^\infty\xi^{s-1}\exp (-\pi\xi \clR_a(\ga))d\xi=
\frac{\G(s)}{\pi^s}\f1{ (\clR_a(\ga))^{s}}\,.
\eq

 It follows from (\ref{ez}) that
\beq{mte}
\clM(\te_a))(s)=\frac{\G(s)}{\pi^s}\zeta_a(s)\,,~a\neq 0\,.
\eq
The theta series $\te_0$ (\ref{th}) have the constant term.
To define the component $\zeta_0(s)$ we shift the theta function $\te_0(\xi)-1$.
 Its Mellin transform
defines the Epstein zeta function
\beq{ep0}
\clM(\te_0-1)(s)=\frac{\G(s)}{\pi^s}\zeta_0(s)\,.
\eq
It means that we pass to the vector
$\bm{\ti\Psi}(\xi)$ corresponding to the vector
$\bm{\Psi}(\tau)$ (\ref{mp1})
\beq{mp}
\bm{\ti\Psi}(\xi)=\bm{\Theta}(\xi)-e_0\,.
\eq

As we mentioned earlier, the substraction of $e_0$ is not compatible with the equivariance.
The $S$-equivariance of $\bm{\Theta}(\xi)$ implies (see(\ref{sac})) that
 \beq{sp}
\bm{\ti\Psi}(\xi)=
A(S|\xi)\rho_l(\bfs)\bm{\ti\Psi}(1/\xi)+A(S|\xi)\rho_l(\bfs)e_0-e_0\,.
\eq

It follows from (\ref{mte}) and (\ref{ep0}) that the  Mellin transform of $\bm{\ti\Psi}(\xi)$
is the vector $\bm{\Xi}(s)$
(\ref{mz1})
\beq{pr}
\bm{\Xi}(s)=\clM(\bm{\ti\Psi})(s)=\int_0^\infty\xi^{s-1}
\bm{\ti\Psi}_l(\xi)d\xi\,.
\eq



\subsubsection{C-invariant $ \bm{\Xi}_l(s)$-vectors }

 The vector $ \bm{\zeta}_l(s)$ is C-invariant if
\beq{cks}
\widetilde{\pi^\vee_k}(C)\bm{\zeta}_l(s)=\bm{\zeta}_l(s)\,,
\eq
where the operator $\widetilde{\pi^\vee_k}(C)$ is defined by (\ref{me3}).
Since this operator as well $\pi^\vee_k(C)$ depends only on $k$ and $\rho_l(\bfc)$
the C-invariant vectors  $\bm{\zeta}_l(s)$ have the same form as C-invariant theta vectors
(\ref{sl1})--(\ref{se66}). The C-invariant set of the Epstein vector zeta functions takes
the form
\vspace{-0.5cm}
\paragraph{A$_{n}$, n=4m }.\\
\beq{sz1}
\bm{\zeta}_{{\rm sl}(4m+1)}(s)=(\zeta_0,\zeta_1,\ldots,\zeta_{2m},\zeta_{2m}\ldots\zeta_1)\,.
\eq
\vspace{-0.8cm}
\paragraph{A$_{n}$, n=4m+2}.\\
\beq{sz2}
\bm{\zeta}_{{\rm sl}(4m+3)}(s)=(0,\zeta_1,\ldots,\zeta_{2m+1},-\zeta_{2m+1},\ldots,-\zeta_1)\,.
\eq
\vspace{-0.8cm}
\paragraph{D$_{n}$, n=2m}.\\
\beq{so12}
\bm{\zeta}_{{\rm so}(4m)}(s)=(\zeta_0,\zeta_1,\zeta_{2},\zeta_{3})\,.
\eq
\vspace{-0.8cm}
\paragraph{E$_6$}.\\
\beq{se61}
\bm{\zeta}^C_{{\rm E}_6}(s)=(0,\zeta_1,-\zeta_1)\,.
\eq
\vspace{-0.8cm}
\paragraph{E$_8$}.\\
\beq{se6}
\bm{\zeta}_{{\rm E}_8}(s)=\zeta_0\,.
\eq


\subsection{Functional equation}

Consider the distribution $e_0/s\in\widetilde{S}^\vee_l$.
Define the shifted  zeta vector
\beq{mz}
\bm{\widehat{\Xi}}^C(s,k)=\bm{\Xi}^C(s,k)-\bfe(-k/2)\rho_l(S)\frac{e_0}s
\eq


\begin{predl}
The   vector $\bm{\widehat\Xi}(s,k)$
satisfies the  matrix functional equation
\beq{fe}
\fbox{$\bm{\widehat\Xi}(2k-s,k)=\widetilde{\pi^\vee_k}(S)\bm{\widehat\Xi}(s,k)$}\,.
\eq
\end{predl}
The proof, with minor modifications, reproduces the proof of the functional equation for the Riemann zeta function.
First we formulate two Lemmas.
\begin{lem}
 Let
\beq{gsk}
G(s,k|\xi)=(\xi^{k-s}\bfe(-k/2)\rho_l(S)+\xi^s)\bm{\ti\Psi}^C(\xi)\,,~~\bfe(x)=e^{\pi\imath x}\,,
\eq
where $\bm{\ti\Psi}^C(\xi)$ is defined in (\ref{mp}).
Then $G(s,k|\xi)$ satisfies the equation
$$
G(s,k|\xi)=\bfe(-k/2)\rho_l(S)G(k-s,k|\xi)+\xi^{s}(\pi_k(C)e_0-e_0)\,.
$$
\end{lem}
\emph{Proof}\\
Let us apply the operator $\bfe(-k/2)\rho_l(S)$ to $G(k-s,k|\xi)$
\beq{1q}
\bfe(-k/2)\rho_l(S)G(k-s,k|\xi)=\xi^{s}\bfe(-k)\rho^2_l(S)\bm{\ti\Psi}^C(\xi)+
\bfe(-k/2)\xi^{k-s}\rho_l(S)\bm{\ti\Psi}^C(\xi)\,.
\eq
Consider the first term
$$
\xi^{s}\bfe(-k)\rho^2_l(S)\bm{\ti\Psi}^C(\xi)=\xi^{s}\widetilde{\pi^\vee_k}(C)\bm{\ti\Psi}^C(\xi)\,.
$$
The $C$-invariance of $\bm{\Theta}^C(\xi)$ implies
that $\pi^\vee_k(C)\bm{\ti\Psi}^C(\xi)=\bm{\ti\Psi}^C(\xi)-(\pi^\vee_k(C)e_0-e_0)$.
It means that
$$
\xi^{s}\bfe(-k)\rho^2_l(S)\bm{\ti\Psi}^C_l(\xi)=\xi^{s}\left(\bm{\ti\Psi}^C(\xi)-
(\pi^\vee_k(C)e_0-e_0)\right)\,.
$$
Substitute this expression in (\ref{1q})
$$
\bfe(-k/2)\rho_l(S)G(k-s,k|\xi)=\xi^{s}\left(\bm{\ti\Psi}^C(\xi)-(\pi^\vee_k(C)e_0-e_0)\right)+
\bfe(-k/2)\xi^{k-s}\rho_l(S)\Psi^C(\xi)=
$$
$$
\bfe(-k/2)\xi^{k-s}\rho_l(S)\bm{\ti\Psi}^C(\xi)+\xi^{s}\bm{\ti\Psi}^C(\xi)-\xi^{s}(\pi^\vee_k(C)e_0-e_0)=
$$
$$
G(s,k|\xi)+\xi^{s}(\pi^\vee_k(C)e_0-e_0)\,.
$$
$\Box$

\begin{lem}
Let $E(s,k)$
\beq{e4}
E(s,k)=\frac{\bfe(- k/2)}{s-k}\rho_l(\bfs) e_0-\f1{s}e_0=
(\ti\pi^\vee_k(S)-Id_l)\frac{e_0}s\,.
\eq
Then it satisfies the equation
\beq{e5}
E(s,k)-\bfe(-k/2)\rho_l(S)E(k-s,k)=(\ti\pi^\vee_k(C)-Id_l)\frac{e_0}s\,.
\eq
\end{lem}
\emph{Proof}\\
The second term in LHS of (\ref{e5}) is equal
$$
\bfe(-k/2)\rho_l(S)E(k-s,k)=-\bfe(-k)\rho^2_l(S)\frac{e_0}s+\bfe(-k/2)\rho_l(S)\frac{e_0}{s-k}=
$$
$$
-\ti\pi^\vee_k(C)\frac{e_0}s+\bfe(-k/2)\rho_l(S)\frac{e_0}{s-k}\,.
$$
From this equality, the lemma immediately follows.
$\Box$

\paragraph{Proof of functional equation}.\\

For generality, take $k$ to be an arbitrary integer.
Define $\bm{\Xi}^C_l(s)$ by the integral(\ref{pr})
$$
\bm{\Xi}^C(s)=\clM(\bm{\ti\Psi}^C)(s)=\int_0^\infty\xi^{s}
\bm{\ti\Psi}^C_l(\xi)d\xi/\xi\,,
$$
where $\bm{\ti\Psi}^C(\xi)$  is defined in (\ref{mp}.
Represent this integral
 as the sum
\beq{in1}
\bm{\Xi}(s)=\underbrace{\int_0^1\xi^{s-1}\bm{\ti\Psi}^C(\xi)d\xi}_{\underline{1}}
+\underbrace{\int_1^\infty\xi^{s-1}\bm{\ti\Psi}^C(\xi)d\xi}_{\underline{2}}\,.
\eq
Substitute in the first integral the representation (\ref{sp}) for $\bm{\ti\Psi}^C_l(\xi)$
and represent the first integral as the sum
$$
I_{\underline{1}}=I_{\underline{3}}+I_{\underline{4}}\,,
$$
where
$$
I_{\underline{3}}=\int_0^1\xi^{s-1}A(S|\xi)\rho_l(\bfs)\bm{\ti\Psi}^C_l(1/\xi)d\xi
$$
\beq{i4}
I_{\underline{4}}=\int_0^1\xi^{s}(A(S|\xi)\rho_l(\bfs) e_0 -e_0)d\xi/\xi
\eq
Taking the integral $I_{\underline{4}}$ we find that it is equal to $E(s,k)$ (\ref{e4}).

Rewrite (\ref{in1}) as
\beq{xin}
\bm{\Xi}^C(s)=\bm{\Xi'}^C(s)+I_{\underline{4}}\,,
\eq
where
\beq{si}
 \bm{\Xi'}^C(s)=I_{\underline{2}}+I_{\underline{3}}\,,
 \eq
Rewrite the integral $I_{\underline{3}}$ as
$$
I_{\underline{3}}=e^{-\pi^\vee\imath k/2}\int_0^1\xi^{s-1 -k}\rho_l(\bfs)\bm{\ti\Psi}(1/\xi)d\xi
=e^{-\pi^\vee\imath k/2}\int_1^\infty\xi^{k-s-1}\rho_l(\bfs)\bm{\ti\Psi}(\xi)d\xi\,.
$$
 Then the sum (\ref{si}) takes the form
\beq{xi1}
\bm{\Xi'}^C(s)=\int_1^\infty\left(\xi^{k-s}e^{-\pi^\vee\imath k/2}\rho_l(\bfs) +\xi^{s}\right)\bm{\ti\Psi}(\xi)d\xi/\xi\stackrel{(\ref{gsk})}{=}\int_1^\infty
G(s,k|\xi)d\xi/\xi\,.
\eq
and
$$
\bfe(-k/2)\rho_l(S)\bm{\Xi'}^C(k-s)=\int_1^\infty
\bfe(-k/2)\rho_l(S)G(k-s,k|\xi)d\xi/\xi\,.
$$

Using Lemma 4.1 we obtain
$$
\bm{\Xi'}^C(s)-\bfe(-k/2)\rho_l(S)\bm{\Xi'}^C(k-s)
=(\ti\pi^\vee_k(C)e_0-e_0)\int_1^\infty\xi^{s}d\xi/\xi\,.
\footnote{ Here we exploit the fact that the operator $\pi^\vee_k(C)$ does not depend on $\xi$ and therefore coincides
with the operator $\ti\pi^\vee_k(C)$ (see (\ref{cac}) and (\ref{me3})).}
$$

Thus,
\beq{ks}
\bm{\Xi'}^C(s)-\bfe(-k/2)\rho_l(S)\bm{\Xi'}^C(k-s)=(\ti\pi^\vee_k(C)-Id_l)\frac{e_0}s\,.
\eq

Consider the integral $I_{\underline{4}}(s)=E(s,k)$.
Due to Lemma 4.2
$$
I_{\underline{4}}(s)-\bfe(-k/2)\rho_l(S)I_{\underline{4}}(k-s)=(\ti\pi^\vee_k(C)-Id_l)\frac{e_0}s
$$
From (\ref{xin}) we find that
\beq{k1}
\bm{\Xi}^C(s)-\bfe(-k/2)\rho_l(S)\bm{\Xi}^C(k-s)=(\ti\pi^\vee_k(C)-Id_l)\frac{e_0}s\,.
\eq

For
$$
\bm{\hat\Xi}^C(s)=\bm{\Xi}^C(s)-\bfe(-k/2)\rho_l(S)\frac{e_0}s
$$
the functional equation (\ref{k1}) takes the form
$$
\bm{\hat\Xi}^C(s)-\bfe(-k/2)\rho_l(S)\bm{\hat\Xi}^C(k-s)=0\,.
$$
Due to (\ref{me2}) this equation is equivalent to (\ref{fe}).$\Box$

\section{Concluding remark}

This construction of the Epstein vector zeta functions, including the functional equation, is
represented by the diagram (\ref{01}). It is likely that it can be
 extended to a broader class of lattices, that includes the ADE lattices as a special case. These are even lattices equipped with a positive definite quadratic form for which the discriminant group is finite. We focus on the ADE lattices here because in this case the results can be made more intuitive 
 and it is possible to provide an explicit classification of lattices when the functional equation is satisfied.

{\small {\bf Acknowledgments}\\
I would like to thank V. Gritsenko for his interest in the work and the productive discussion.
The work was supported by Russian Science Foundation
grant  24-12-00178.
}

\end{document}